\documentclass[a4paper,12pt]{eptcs}
\usepackage{graphicx}
\usepackage{tikz}
\usetikzlibrary{er,positioning,bayesnet,matrix,intersections}
\usepackage{hyperref}
\usepackage{amsmath,amssymb}
\pgfdeclarelayer{fg}
\pgfdeclarelayer{crossing over}
\pgfsetlayers{main,crossing over,fg}
\tikzset{label/.style={font=\scriptsize}}
\usepackage{float}
\usepackage{amsfonts}
\usepackage{braket}
\usepackage{xspace}
\usepackage{tikz-cd}
\tikzset{widthone/.style={draw, minimum width=0.6cm, fill=white, minimum height=16pt, inner sep=-10pt}}
\tikzset{widthtwo/.style={draw, minimum width=1.6cm, fill=white, minimum height=16pt, inner sep=-10pt}}
\tikzset{widththree/.style={draw, minimum width=2.2cm, fill=white}}
\tikzset{label/.style={font=\scriptsize}}
\usepackage{amsthm}
\theoremstyle{definition}
\newtheorem{definition}{Definition}[section]
\newtheorem{example}[definition]{Example}

\theoremstyle{satz}
\newtheorem{satz}[definition]{Proposition}

\theoremstyle{theorem}
\newtheorem{theorem}[definition]{Theorem}

\def\FHilb{\ensuremath{\mathbf{FHilb}}\xspace}
\def\QDesign{\ensuremath{\mathbf{QDesign}}\xspace}

\theoremstyle{corollary}

\theoremstyle{lemma}
\newtheorem{lemma}[definition]{Lemma}

\newcounter{remark}[section]

\usepackage{graphicx}
\usepackage{verbatim}
\tikzset{dot/.style={draw, fill, circle, inner sep=0pt, minimum width=5pt}}

\newcommand\ignore[1]{}

\newcommand\vc[1]{\begin{tabular}{ccc}#1\end{tabular}}

\def\C{\ensuremath{\mathcal{C}}}

\def\Tr{\ensuremath{\mathrm{Tr}}}
\def\N{\mathbb{N}}
\def\R{\mathbb{R}}
\def\I{\mathbb{I}}

\newcounter{commcounter}
\def\titlerunning{A Categorical Model for Classical and Quantum Block Designs}
\def\authorrunning{Paulina Goedicke and Jamie Vicary}

\begin{document}

\title{\bf A Categorical Model for Classical\\and Quantum Block Designs}
\author{\small\vc{\vc{Paulina L. A. Goedicke\\Institute for Theoretical Physics\\University of Cologne\\\texttt{goedicke@thp.uni-koeln.de}} & \vc{Jamie Vicary\\Department of Computer Science\\University of Cambridge\\\texttt{jamie.vicary@cl.cam.ac.uk}}}
}
\date{} 

\maketitle

\begin{abstract}
Classical block designs are important combinatorial structures with a wide range of applications in Computer Science and Statistics. Here we give a new abstract description of block designs based on the arrow category construction. We show that models of this structure in the category of matrices and natural numbers recover the traditional classical combinatorial objects, while models in the category of completely positive maps yield a new definition of quantum designs. We show that this generalizes both a previous notion of quantum designs given by Zauner and the traditional description of block designs. Furthermore, we demonstrate that there exists a functor which relates every categorical block design to a quantum one.
\end{abstract}

\section{Introduction}\label{sec. intro}

Combinatorial design theory has a variety of applications in computer science and statistics. Its main focus lies on finite discrete structures, such as (hyper)graphs, finite projective and affine planes, block designs, orthogonal arrays and Latin squares that fulfil certain constraints on the arrangement of elements ~\cite{stinson}. With the development of quantum computing, it is natural to ask how combinatorial objects fit into quantum information theory. Many interesting intersections between these fields are already known; for example, it has been shown that the problem of constructing mutually orthogonal Latin squares of order $d$ is equivalent to constructing a 2-uniform state of $N$ qudits of $d$ levels having $d^{2}$ positive terms~\cite{PhysRevA.92.032316}.\par
Some combinatorial objects have an analogous quantum form, which can lead to rich insights. For instance, Vicary et al. have introduced the notion of \textit{quantum Latin squares} (QLS), a quantum analogue of Latin squares with which one can build a new construction scheme for unitary error bases~\cite{musto2016quantum}. Based on this, Goyeneche et al. introduced \textit{quantum orthogonal arrays} and showed that they are related to QLSs, in the same way that orthogonal arrays are related to Latin squares~\cite{Goyeneche_2018},~\cite{Goy}. Quantum Latin squares also play a role in the classification of other biunitary constructions such as Hadamard matrices~\cite{reutter2019biunitary}, quantum teleportation and error correction~\cite{musto2016quantum}, and also give a new construction scheme for mutually unbiased bases (MUBs) as shown by Musto~\cite{Musto_2017}. Only recently, {\.{Z}}yczkowski et al. have found a quantum solution to the famous Euler's problem of thirty six officers, i.e. classically no two orthogonal Latin squares of order six exist, by constructing two orthogonal QLS of order six~\cite{Rather_2022}.
Other examples of the connection between classical and quantum combinatorics include the work by Wocjan and Beth on mutually unbiased basis construction from orthogonal Latin squares~\cite{wocjan}, and Wooter's description of affine plane constructions of mutually unbiased bases~\cite{wooters}. However, a unified perspective on the relationship between classical and quantum combinatorics remains elusive.

A prominent example of combinatorial designs are \emph{balanced incomplete block designs} (BIBDs)~\cite{stinson}, combinatorial designs with special balance constraints on the arrangement of elements, and Zauner has extended this to give a more general notion of quantum design~\cite{Zauner}. In this paper we develop a category-theoretical model for both classical and quantum designs, using the language of arrow categories. We will start by reviewing classical design theory and Zauner's notion of quantum designs. We then develop a categorical framework based on arrow categories, that transfers the intrinsic properties of block designs to a pointed monoidal dagger category. Applied to the categories $\mathbf{Mat}(\mathbb{N})$ and $\mathbf{CP}[\mathbf{FHilb}]$, this framework yields a categorical model for both block designs and quantum designs. This not only leads to a more general description of both classical and quantum designs via completely positive maps, but also allows us to relate classical to quantum designs via a functor. Moreover, we will use these techniques to define a category of mutually-unbiased bases. 

 
\subsection{Structure of the paper}
This paper is organised as follows. In Section~\ref{sec. background} and Section~\ref{sec. category theory} we will summarise the central mathematical concepts that are used in the paper, covering elementary combinatorics, quantum designs as they were defined by Zauner, and elements from category theory. Moreover, we give the definition of a category of BIBDs, namely $\mathbf{Block}$ and a category of general designs, namely $\mathbf{Design}$. In Section~\ref{sec. design construction} we will develop the \emph{design construction} that describes the design structures discussed in Section \ref{sec. bdesigns} via an abstract operation on a pointed monoidal dagger category. By applying this construction to $\mathbf{Mat}(\mathbb{N})$, we then show in the beginning of Section~\ref{sec. results} that this recovers the category theoretical description of classical designs, namely $\mathbf{Design}[\mathbf{Mat}(\mathbb{N})]$, and we show that there exists a functor from the subcategory $\mathbf{BDesign}[\mathbf{Mat}(\mathbb{N})]$ to $\mathbf{Block}$. Furthermore, we define a category $\mathbf{QDesign}$ of quantum designs by applying the design construction to $\mathbf{CP[FHilb]}$, and show that it contains a subcategory where objects are uniform and regular quantum designs of degree 1. This also allows us to define a category of MUBs. Finally, we construct a functor between $\mathbf{BDesign}[\mathbf{Mat}(\mathbb{N}), \mathbf{FSet}]$ and $\mathbf{QDesign}$, yielding a relation between categorical block designs and quantum design in accordance with earlier results given by Zauner~\cite{Zauner}.

Throughout the paper we will only consider finite dimensional Hilbert spaces.
\subsection{Open questions}
It would be interesting to know if the more general quantum designs that we define have an application in quantum computing, and in particular if the trace-preserving examples are interesting to use as quantum channels. The same holds for the CP-maps representing classical designs, which can be interpreted in terms of statistical mechanics.\footnote{This is because they are morphisms of $\mathbf{CP_{c}}[\mathbf{FHilb}]$ (see Section~\ref{sec. category theory}).} Moreover, a category theoretical perspective might lead to new insights on the problem of finding the maximal number of MUBs in arbitrary dimension, a problem which has already been approached by Musto using categorical techniques~\cite{Musto_2017}. Especially in this context, it would be relevant to know how one can embed Latin squares into our model, which may also give new insights on how QLS relate to quantum designs. Finally, one could ask how one can embed mutually unbiased measurements (MUMs) into this framework.

\subsection{Acknowledgements.} 
The first author has been supported by Germany's Excellence Strategy Cluster of Excellence Matter and Light for Quantum Computing (ML4Q) EXC 2004/1 390534769. The second author acknowledges support from the Royal Society. We are grateful to the reviewers for their excellent and insightful comments which improved our paper considerably.

\section{Design Theory}\label{sec. background}
In this section we recall the definition of balanced incomplete block designs, following Stinson ~\cite{stinson} and Buekenhout~\cite{IGeo}. We then review Zauner's definition of quantum designs~\cite{Zauner}.

\def\P{\mathcal{P}}

\subsection{Classical designs}\label{sec. bdesigns}

We begin with the definition of designs.

\begin{definition}
A \emph{design} is given by a set  $V=\{1,..,v\}$ of \emph{points}, and a set $B = \{1, \ldots, b\}$ of \emph{blocks} and an \emph{incidence relation} $I$ between them. 
\end{definition}

\noindent
Designs can also be viewed as bipartite graphs on the partitioned set given by the disjoint union of the blocks and points. More concretely, we can represent a design as an \textit{incidence matrix}, a $v \times b$ matrix $\chi$ with $\chi_{i,j} = 1$ if and only if $(i,j) \in I$, and $\chi_{i,j} = 0$ otherwise. Throughout the paper we will mainly use the incidence matrix representations of designs. 

\begin{definition}
A design $\chi: b \longrightarrow v$ is called 
\begin{itemize}
    \item \textit{k-uniform}, if every block contains exactly $k$ points:
$$\sum_{i=1}^{v} \chi_{i,j} = k, \ \ \text{for all} \ \ j=1,..,b$$
    \item \textit{r-regular}, if every point appears in exactly $r$ blocks:
$$\sum_{j=1}^{b} \chi_{i,j} = r \ \ \text{for all} \ \ i=1,..,v$$
\end{itemize}
\end{definition}
\begin{definition}
    A $k$-uniform and $r$-regular designs is called \textit{$\lambda$-balanced}, if any two points are contained in exactly $\lambda$ blocks. We then have:
$$
\chi \cdot \chi^{T}= \lambda \Bigl(E_{v \times v} - \mathbb{I}_{v \times v} \Bigr) + r\mathbb{I}_{v \times v}
$$
Here $\chi^{T}$ is the transpose incidence matrix, $E_{v \times v}$ denotes the $v \times v$-matrix in which every entry is equal to 1, and $\mathbb{I}_{v \times v}$ denotes the $v \times v$ identity matrix.
\end{definition}
\noindent
 The last expression from above means that by multiplying $\chi$ with its transpose one obtains a matrix where every off-diagonal entry is equal to $\lambda$ and every diagonal entry is equal to $r$.
These properties combine to give the important notion of block design.
\begin{definition}\label{def. 2 design}
A \emph{block design}, or a  \emph{$(v,k,r,b,\lambda)$-design}, is a design $\chi: b \longrightarrow v$ which is $k$-uniform, $r$\-regular and $\lambda$-balanced.\footnote{What we define here is actually known as a balanced incomplete block design (BIBD) in the literature, whereas block designs define a more general concept. For the sake of simplicity we will call BIBD's block designs.}
\end{definition}
By simple counting arguments, one can easily derive the following equational properties~(\cite{stinson}, p. 4-5):
\begin{lemma}
\label{lem:blockproperties}
For a $(v,k,r,b,\lambda)$-design, the following equations hold:
\begin{align}
    b\cdot k &= r\cdot v \label{bd1}\\
    \lambda (v-1)&=r(k-1) \label{bd2}
\end{align}
\end{lemma}


\begin{definition}
\label{def. symmetric design}
A $(v,k,r,b,\lambda)$-design is \textit{symmetric} when $v=b$; that is, when there are as many points as blocks. Lemma~\ref{lem:blockproperties} then implies that $r=k$.
\end{definition}
\begin{example}[\cite{stinson}, p.27]
Consider a finite projective plane of order $d$. We then have $v=d^{2}+d+1$ points and $b=d^{2}+d+1$ lines such that there are $k=d+1$ points on each line and each point appears on $r=d+1$ lines. Moreover, every pair of lines intersect in exactly one point. Hence we have a symmetric block design with parameters  $v=b=d^{2}+d+1$, $r=k=d+1$ and $\lambda=1$. 
\end{example}
We now consider the appropriate notion of homomorphism of block design.
\begin{definition}
Consider two designs $\chi: b \longrightarrow v $ and $\chi': b' \longrightarrow v'$. A \textit{design homomorphism} $f:\chi \to \chi'$ is a pair of functions $f_v:v \to v'$ and $f_b:b \to b'$ such that the following diagram commutes:
\[\begin{tikzcd}
b\arrow{d}{\chi} \arrow{r}{f_b} & b' \arrow{d}{\chi'}\\
v \arrow{r}{f_v} & v'
\end{tikzcd}
\]

\end{definition}

\noindent
We can use this to obtain categories of designs and block designs, as follows.
\begin{definition}\label{def. category BDesign}
The category $\mathbf{Design}$ has designs as objects, and design homomorphisms as morphisms. The category $\mathbf{Block}$ is the full subcategory on the block designs.
\end{definition}

\noindent
This definition could alternatively be given in terms of points and blocks, but we define it in this abstract way to better relate the categorical machinery to follow. A related definition was given by D\"orfler and Waller  \cite{Drfler1980ACA} who explored categories of hypergraphs. Since hypergraphs generalise relations, designs are instances of hypergraphs. In their definition they use the power-set functor to assign to each edge (block) the set of vertices (points) it is incident with. For our notion of a block design we do not make explicit use of the power-set functor; however, this could be an interesting approach since the category of relations is the Kleisli category of the power-set functor, and we would like to consider it in future work.

\subsection{Quantum designs}\label{sec. qdesigns}

A notion of quantum design has been presented by Zauner in his PhD thesis~\cite{Zauner}. Here we recall that definition, adapting the terminology slightly for consistency.
\begin{definition}\label{def. quantum design}
A \emph{quantum $(v,b)$-design} is a set $\mathit{D}=\{ p_{1},...,p_{v} \}$ of complex orthogonal $b\times b$ projection matrices $p_{i}$ on a $b$-dimensional Hilbert space $\mathbb{C}^{b}$, i.e. $p_{i}=p_{i}^{\dagger}=p_{i}^{2}$ for all $i \in \{1,  \ldots ,v\}$.
\end{definition}

\noindent
As for classical designs above, we introduce certain properties for quantum designs.
\begin{definition}
A quantum $(v,b)$-design is called
\begin{itemize}
    \item \emph{$r$-regular} if there exists some $r \in \N$ with $\Tr(p_i) = r$ for all $i \in \{1, \ldots, v\}$;
    \item \emph{$k$-uniform} if there exists some $k \in \R$ with $\sum_{i=1}^v p_i = k \cdot \I_{b \times b}$.
\end{itemize}
\end{definition}

\noindent
\begin{definition}
Given a quantum $(v,b)$-design, its \textit{degree} is the cardinality of the set $\{\Tr(p_i p_j)| i,j \in \{ 1, \ldots, v\}, i \neq j\}$.
\end{definition}

\noindent
It follows that a quantum design has degree 1 just when there exists some $\lambda \in \mathbb R$ such that 
\begin{align}
    \Tr(p_{i}p_{j})=\lambda && \forall i,j=1,...,v && \text{with} && i \neq j.
\end{align}
We will call such a quantum design $\lambda$-\emph{balanced}. That $\lambda$ is real in this case follows from a simple argument: $\lambda = \Tr(p_i p_j) = \Tr((p_i p_j)^\dag)^* = \Tr(p_j^\dag p_i ^\dag)^* = \Tr(p_j p_i)^* = \Tr(p_i p_j)^* = \lambda ^*$.

The following lemma can then be established analogous to Lemma~\ref{lem:blockproperties} for classical designs.
\begin{lemma}\label{lem. properties of quantum designs}
For a $k$-uniform, $r$-regular and $\lambda$-balanced quantum design $D=\{ p_{1},...,p_{v} \}$ with $p_{i} \in \mathbb{C}^{b}$ the following equations hold:
    \begin{align}
        b \cdot k &= v \cdot r \label{eq. 1 qdesign}\\
        \lambda (v-1) &= r(k-1) \label{eq. 2 qdesign}
    \end{align}
\end{lemma}

\noindent
The proof for this lemma can be found in Appendix \ref{app: proofs}.

\begin{definition}\label{def. commutative quantum design} A quantum design is \emph{commutative} when all projection matrices pairwise commute.
\end{definition}
\begin{theorem}[\cite{Zauner}, Theorem 1.10] \label{thrm equivalence designs}
    A commutative quantum design is equivalent to a classical block design.\footnote{This holds because every commutative design is unitarily equivalent to a design comprised of diagonal matrices; as the projections are idempotent, the diagonal entries must therefore be 0 or 1 \cite{Zauner}.}
\end{theorem}

A prominent example of a uniform and regular quantum design with degree 2 are \emph{mutually unbiased bases} (MUBs). These objects play a significant role in quantum information theory.
\begin{example}[MUBs.]
Mutually unbiased bases are a pair of bases $\{ \ket{a_i}\}_{i=0,...,d-1},\{\ket{b_i}\}_{i=0,...,d-1}$ for a $d$-dimensional Hilbert space $H$, such that the inner products $\braket{a_i| b_j}$ are equal for all $i,j=0,...,d-1$. A uniform and regular quantum design of degree 2 with parameters $r=1$, $b=d$, $v=d \cdot k$ and $\Lambda = \{ \frac{1}{d}, 0 \}$ defines a set of $k$ MUB's in a $d$-dimensional Hilbert space $H$. To see that, note that the $d\cdot k$ projectors all have trace one and satisfy the following condition, where $a$ labels the different orthogonal classes, and $i$ labels the projectors within an orthogonal class:
\begin{align}
    \sum_{a=1}^{k} \sum_{i=1}^{d} p_{i}^{a} = k \cdot \mathbb{I}
\end{align}
Moreover, the following holds:
\begin{align*}
    \mathrm{tr}(p^{a}_{i}p^{b}_{j})= \frac{1}{d}(1- \delta_{ab}) + \delta_{ij}\delta_{ab}
\end{align*}
\end{example}
It is easy to see that we get a \emph{complete} set of MUBs if $v$ equals $d \, (d+1)$, as we then have $k=d+1$.\footnote{Complete means that we have $d+1$ MUBs in a $d$-dimensional Hilbert space.} 

\section{Category Theory}\label{sec. category theory}
In this section we give the definition of an arrow category, and explain the CP-construction. We will assume familiarity with basic concepts in category theory and the graphical calculus, and refer to Maclane~\cite{maclane} and Heunen and Vicary~\cite{jamie} for further background.

The categories we will mostly use in this paper are the category of matrices and natural numbers, \textbf{Mat}($\mathbb{N}$), and the category of finite dimensional Hilbert spaces and bounded linear maps, $\mathbf{FHilb}$. 
\begin{example} \label{ex. Rel cat and Hilb cat}
\begin{itemize}
    \item [(i)](\cite{jamie}, p. 16) The category $\mathbf{FHilb}$ has as objects finite dimensional Hilbert spaces and as morphisms bounded linear maps between Hilbert spaces. Composition is the composition of linear maps as ordinary functions and the identity morphisms are given by identity linear maps. The monoidal product is given by the tensor product on Hilbert spaces and the unit object is the one-dimensional Hilbert space $\mathbb{C}$. 
    \item[(ii)] (\cite{IntroCat}, p. 4) The category \textbf{Mat}($\mathbb{N}$) of matrices over $\mathbb{N}$ has objects given by natural numbers. For $m,n \in \mathbb{N}$, the Hom-set $\mathrm{Hom}_{\mathbf{Mat}(\mathbb{N})}(m,n)$ is the set of all $n\times m$-matrices over $\mathbb{N}$, composition being matrix multiplication. The monoidal product on objects is given by the multiplication of numbers and on morphisms by the Kronecker product of matrices. The monoidal unit is the natural number~$1$.
\end{itemize}
\end{example}

\begin{definition}[See \cite{IntroCat}, p. 23-24]\label{def. arrow category}
For a category \C, its \textit{arrow category} $\mathbf{Arr}[\C]$ is defined as follows:
\begin{itemize}
    \item objects are triples $(A, B, h)$ with $h: A\to B$ in \C;
    \item morphisms $\phi: (A, B, h) \to (A', B', h')$ are pairs of morphisms $\phi_{A}: A \to A'$ and $\phi_{B}: B \to B'$ in $\mathcal{C}$ such that the following diagram commutes:
    \[\begin{tikzcd}
    A \arrow[swap]{d}{h} \arrow{r}{\phi_A} &     \arrow{d}{h'} A' \\
    B \arrow{r}{\phi_B} & B'
\end{tikzcd}
\]
\end{itemize}
\end{definition}

\subsection{The CP construction}
The concept of completely positive maps is well-established~\cite{murphystaralgebras}. Here we use Selinger's categorical description of completely positive maps~\cite{selinger}, as follows, exploiting the notion of dagger Frobenius structure, which is standard in the categorical quantum mechanics literature~\cite{jamie}.
\begin{definition} \label{def. CP condition}
In a monoidal dagger category,  let $(A, \mu_{A}, \eta_{A})$ and $(B, \mu_{B}, \eta_{B})$ be dagger Frobenius structures.\footnote{A dagger Frobenius structure is a monoid structure which, together with its dagger, satisfies a Frobenius condition. For an introduction to dagger Frobenius structures we refer to Heunen and Vicary~\cite{jamie}.} A morphism $f: A \to B$ satisfies the \emph{CP-condition} if there exists some object $X$ and some morphism $g: A \otimes B \to X$ such that the following equation holds:
$$
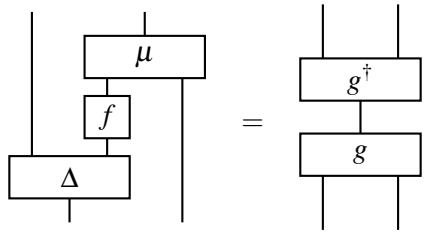
\begin{aligned}
\begin{tikzpicture}[scale=1, thick, yscale=.8]
\draw (1,-.75) to +(0,1);
\draw (0.5,0) to +(0,2.75);
\draw (1.5,0) to +(0,1.75);
\draw (2,2) to +(0,.75);
\draw (2.5, -.75) to +(0, 3);
\node [widthtwo] at (1,0) {$\Delta$};
\node [widthtwo] at (2,2) {$\mu$};
\node [widthone] at (1.5,1) {$f$};
\end{tikzpicture}
\end{aligned}
\quad=\quad
\begin{aligned}
\begin{tikzpicture}[scale=1, thick]
\draw (6,0) to +(0,1);
\draw (7,0) to +(0,1);
\draw (6.5,1) to +(0,1);
\draw (6,2) to +(0,1);
\draw (7,2) to +(0,1);
\node [widthtwo] at (6.5, 1) {$g$};
\node [widthtwo] at (6.5,2) {$g^\dagger$};
\end{tikzpicture}
\end{aligned}
$$ 
\end{definition}

\noindent
One can show that, in a symmetric monoidal dagger category, a morphism that satisfies this condition constitutes a CP-map \cite{jamie}.
\begin{example}\label{ex. POVM CPM}%
In \FHilb, consider a POVM consisting of $b$ projections $p_{i}: H \to H$. One can define a completely positive map $\varphi: \mathbb{C}^{b} \to H \otimes H^{*}$ that sends the computational basis vector $\ket{i}$ to $p_{i}$. Graphically, we can represent $\varphi$ as follows, where $b_{H}: \mathbb{C} \to H^{*} \otimes H$ is the evaluation map:
$$
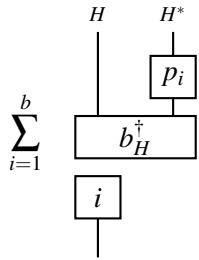
\sum_{i=1}^{b}
\quad
\begin{aligned}
\begin{tikzpicture}[scale=1, thick, yscale=.8]
\draw (2,1) to +(0,1);
\draw (2,3) to (2,4.75) node [above, label] {$H$} ;
\draw (3,3)  to (3,4.75) node [above, label] {$H^{*}$};
\node [widthone] at (3,4) {$p_{i}$};
\node [widthtwo] at (2.5,3) {$b_{H}^\dagger$};
\node [widthone] at (2,2) {$i$};
\end{tikzpicture}
\end{aligned}
$$
\end{example}

\begin{satz} \label{prop. CP construction}
Let $(\mathcal{C}, \otimes_{\mathcal{C}}, \mathbb{I}_{\mathcal{C}})$ be a monoidal dagger category. There is a category $\mathbf{CP}[\mathcal{C}]$ in which
\begin{itemize}
    \item  objects are special symmetric dagger Frobenius structures in $\mathcal{C}$
    \item morphisms are morphisms of $\mathcal{C}$ that satisfy the CP-condition.
\end{itemize}
\end{satz}
\begin{example}  \label{ex. FHilb}
In $\mathbf{CP}[\mathbf{FHilb}]$ objects are finite dimensional $H^{*}$-algebras, i. e. an algebra A that is also a Hilbert space with an anti-linear involution $\dagger: A \to A$ satisfying $\langle ab| c\rangle = \langle b| a^{\dagger} c\rangle = \langle a| c b^{\dagger} \rangle$, and morphisms are completely positive maps. 
\end{example}
\begin{satz}
   The category $\mathbf{CP_{c}}[\mathcal{C}]$ with \emph{classical structures}, i.e. special commutative dagger Frobenius structures in $\mathcal{C}$, as objects and completely positive maps between these structures as morphisms, is a subcategory of $\mathbf{CP}[\mathcal{C}]$. 
\end{satz}
\begin{satz}[\cite{jamie}, p. 241]\label{prop. equivalence CPc FHilb and MatN}
The category $\mathbf{CP_{c}}[\mathbf{FHilb}]$ is monoidally equivalent to $\mathbf{Mat}(\mathbb{N})$. 
\end{satz}

\noindent
An interpretation of these constructions is the following: $\mathcal{C}$ models pure state quantum mechanics, $\mathbf{CP}[\mathcal{C}]$ models mixed state quantum mechanics, while  $\mathbf{CP_{c}}[\mathcal{C}]$ describes statistical mechanics~\cite{jamie}. 

The final piece of structure we require is that of pointed monoidal category.
\begin{definition}
A \emph{pointed} monoidal category is a monoidal category for which every object $A$ is equipped with a canonical morphism $p_{A}: \mathbb{I} \to A$.
\end{definition}

\begin{example} We obtain examples as follows from the categories we have been considering:\label{ex. pointed structures}
\begin{itemize}
    \item [(i)] The category $\mathbf{CP}[\mathcal{C}]$ has a pointed structure given by the adjoint of the trace map $V \otimes V^{*} \to \mathbb{I}$.
    \item[(ii)]  In \textbf{Mat}($\mathbb{N}$) a pointed structure is given by a column matrix with a $1$ at every entry: $p_{n}: 1 \to n$.
\end{itemize}
\end{example}
\section{Categorical Block Designs}\label{sec. design construction}
In this section we will develop a construction that gives an abstract notion of the uniformity-, regularity- and $\lambda$-balanced condition from Section \ref{sec. bdesigns} in an arbitrary rigid monoidal category. We will call this the \emph{design construction}.
\begin{definition}[Design construction] \label{def. Design construction}
    Let $F:\mathcal{D} \hookrightarrow \mathcal C$ be a faithful monoidal functor between pointed monoidal dagger categories. The category $\mathbf{Design}[\mathcal{C}, \mathcal D]$ is the subcategory of $\mathbf{Arr}[\mathcal{C}]$ where the morphisms are given by pairs of morphisms of $\mathcal{C}$ which are in the image of the functor $F$; we omit $F$\ from the notation, ensuring it is clear from the context. Where $F=\mathrm{id}$, we simply write $\mathbf{Design}[\mathcal C]$. 
\end{definition}
\begin{definition}\label{def. RUDesign construction}
The category $\mathbf{RUDesign}[\mathcal{C}, \mathcal D]$ is the subcategory of  $\mathbf{Design}[\mathcal{C}, \mathcal D]$ where objects \mbox{$f: A \to D$} are $r$-regular and $k$-uniform, for scalars $r$, $k$ $\in$ $\mathrm{Hom}(\mathbb{I}_{\mathcal{C}},\mathbb{I}_{\mathcal{C}} )$, with the pointed structure and its dagger represented by a black dot:
\begin{align*}
    \begin{aligned}
    \begin{tikzpicture}[scale=.6, thick, yscale=.8]
    \draw (1,0) to (1,3) node [above, label] {$D$};
    \node [widthone] at (1,1.5) {$f$};
    \node [dot] at (1,0) {};
    \end{tikzpicture}
    \end{aligned}
    \quad&=\quad r\,\,\,
    \begin{aligned}
        \begin{tikzpicture}[scale=.6, thick, yscale=.8]
    \draw (0,0) to (0,3) node [above, label] {$D$} ;
    \node [dot] at (0,0){};
    \end{tikzpicture}
    \end{aligned}
    &
    \begin{aligned}
    \begin{tikzpicture}[scale=.6, thick, yscale=.8]
    \draw (0,0) node [below, label] {$A$} to (0,3);
    \node [widthone] at (0,1.5) {$f$};
    \node [dot] at (0,3){};
    \end{tikzpicture}
    \end{aligned}
    \quad&=\quad k\,\,\,
    \begin{aligned}
        \begin{tikzpicture}[scale=.6, thick, yscale=.8]
    \draw (0,0) node [below, label] {$A$}  to (0,3);
    \node [dot] at (0,3){};
    \end{tikzpicture}
    \end{aligned}
\end{align*}
\end{definition}
\begin{lemma}\label{lemma: bd1 in Design(C)}
In $\mathbf{RUDesign}[\mathcal{C}, \mathcal D]$ for any $k$-uniform, $r$-regular object $f: A \to D$, the following equations hold:
\begin{align}
    k \cdot \mathrm{dim}(D) = r \cdot \mathrm{dim}(A) \label{eq. bd1 cat}
\end{align}
where $\mathrm{dim}(A)= p_{A}^{\dagger}\circ p_{A}$ for $A$ $\in$ $\mathrm{obj}(\mathcal{C})$. Here $p_{A}: 1 \to A$ is the pointed structure of $\C$.
\end{lemma}
\begin{proof}
Via composition with $p^{\dagger}_{D}$ and $p_{A}$ respectively, the regularity and uniformity condition become:
\begin{align*}
\begin{aligned}
\begin{tikzpicture}[scale=1, thick, yscale=.8]
\draw (0,0) to +(0,2);
\node [widthone] at (0,1) {$f$};
\node [dot] at (0,0) {};
\node [dot] at (0,2) {};
\end{tikzpicture}
\end{aligned}
\quad&=\quad k \cdot \dim(A)
&
\begin{aligned}
\begin{tikzpicture}[scale=1, thick, yscale=.8]
\draw (0,0) to +(0,2);
\node [widthone] at (0,1) {$f$};
\node [dot] at (0,0) {};
\node [dot] at (0,2) {};
\end{tikzpicture}
\end{aligned}
\quad&=\quad r \cdot \dim(D)
\end{align*}
Hence Eq. \ref{eq. bd1 cat} holds.
\end{proof}
\begin{definition} \label{def. BDesign construction}

    The category $\mathbf{BDesign}[\mathcal{C}, \mathcal D]$ is the subcategory of $\mathbf{RUDesign}[\mathcal{C}, \mathcal D]$ where all $k$\-uniform and $r$-regular objects $f: A \to D$ are $\lambda$-balanced for scalars $\lambda$ $\in$ $\mathrm{Hom}(\mathbb{I}_{\mathcal{C}},\mathbb{I}_{\mathcal{C}} )$:
$$
\begin{aligned}
\begin{tikzpicture}[thick, yscale=.8]
\draw (0,0.25) to +(0,2.5);
\node [widthone] at (0,1) {$f^\dag$};
\node [widthone] at (0,2) {$f$};
\end{tikzpicture}
\end{aligned}
\quad=\quad
\lambda \left(
\,\,\,
\begin{aligned}
\begin{tikzpicture}[thick, yscale=.8]
\draw (0,0.25) to +(0,.75);
\draw (0,2.75) to +(0,-.75);
\node [dot] at (0,1) {};
\node [dot] at (0,2) {};
\end{tikzpicture}
\end{aligned}
\quad-\quad
\begin{aligned}
\begin{tikzpicture}[thick, yscale=.8]
\draw (0,0.25) to +(0,2.5);
\end{tikzpicture}
\end{aligned}
\,\,\,\,
\right)
\quad+\quad
r
\,\,\,
\begin{aligned}
\begin{tikzpicture}[thick, yscale=.8]
\draw (0,0.25) to +(0,2.5);
\end{tikzpicture}
\end{aligned}
$$
\end{definition}

\begin{lemma}\label{lemma: bd2 in BDesign(C)}
In $\mathbf{BDesign}[\mathcal{C}, \mathcal D]$ for any $k$-uniform, $r$-regular object $f: A \to D$, the following equation holds, where $\mathrm{dim}(D)=p_{D}^{\dagger} \circ p_{D}$ for $D$ $\in$ $\mathrm{obj}(\mathcal{C})$:
\begin{align}
    \lambda \cdot (\mathrm{dim}(D) -1) = k \cdot (r-1) \label{eq. bd2 cat}
\end{align}

\end{lemma}
\begin{proof}
To prove Eq. \ref{eq. bd2 cat}, we concatenate the $\lambda$-condition with both $p_{A}$ and $p_{A}^{\dagger}$ which gives:
$$
\begin{aligned}
\begin{tikzpicture}[thick, yscale=.8]
\draw (0,0.25) node [dot] {} to +(0,2.5) node [dot] {};
\node [widthone] at (0,1) {$f^\dag$};
\node [widthone] at (0,2) {$f$};
\end{tikzpicture}
\end{aligned}
\quad=\quad
\lambda (\dim(D)^2 - \dim(D)) + r\, \dim(D)
$$
On the other hand we have:
$$
\begin{aligned}
\begin{tikzpicture}[thick, yscale=.8]
\draw (0,0) node [dot] {} to (0,2.9);
\node [widthone] at (0,1) {$f^\dag$};
\node [widthone] at (0,2) {$f$};
\end{tikzpicture}
\end{aligned}
\quad=\quad
k \,\,\,
\begin{aligned}
\begin{tikzpicture}[thick, yscale=.8]
\draw (0,1) node [dot] {} to (0,2.9);
\node [widthone] at (0,2) {$f$};
\node [dot, white] at (0,0) {};
\end{tikzpicture}
\end{aligned}
\quad=\quad
k\,\,r \,\,\,
\begin{aligned}
\begin{tikzpicture}[thick, yscale=.8]
\draw (0,1) node [dot] {} to (0,2.9);
\node [dot, white] at (0,0) {};
\end{tikzpicture}
\end{aligned}
$$
If we now concatenate with $p_{D}^{\dagger}$, we get:
$$
\begin{aligned}
\begin{tikzpicture}[thick, yscale=.8]
\draw (0,0.1) node [dot] {} to (0,2.9) node [dot] {};
\node [widthone] at (0,1) {$f^\dag$};
\node [widthone] at (0,2) {$f$};
\end{tikzpicture}
\end{aligned}
\quad=\quad
k \,\, r \,\, \dim(D)
$$
From this we can easily deduce Eq. \ref{eq. bd2 cat}. 
\end{proof}

\section{Classical and Quantum Models}\label{sec. results}
In this section we will apply the design-constructions from Section~\ref{sec. design construction} to our model categories $\mathbf{Mat}(\mathbb{N})$ and $\mathbf{CP}[\mathbf{FHilb}]$ and show that this gives us a categorical model of both classical quantum designs.

\subsection{The Category of Block Designs}

Writing $\mathbf{FSet}$ for the category of finite sets and functions, there is a faithful functor $\mathbf{FSet} \hookrightarrow \mathbf{Mat}(\mathbb N)$ which takes every set to the natural number given by its cardinality. Moreover, we have that $\mathbf{FSet} \hookrightarrow \mathbf{Mat}(\mathbb N) \cong \mathbf{CP_c}[\mathbf{FHilb}]$ (see Theorem~\ref{prop. equivalence CPc FHilb and MatN}). In the following we will prove that there exists a functor from the category $\mathbf{BDesign}[\mathbf{Mat}(\mathbb{N}), \mathbf{FSet}]$ to the category $\mathbf{Block}$. Moreover, we will show that $\mathbf{BDesign}[\mathbf{Mat}(\mathbb{N}), \mathbf{FSet}]$ is equivalent to $\mathbf{BDesign}[\mathbf{CP_c}[\mathbf{FHilb}], \mathbf{FSet}]$.
\begin{theorem}
There exists a functor $G:\mathbf{BDesign}[\mathbf{Mat}(\mathbb{N}), \mathbf{FSet}] \longrightarrow \mathbf{Block}$. 
\end{theorem}
\begin{proof}
We first note  that the morphisms in $\mathbf{BDesign}[\mathbf{Mat}(\mathbb{N}), \mathbf{FSet}]$ are given by pairs of functions.
The functor sends each object in $\mathbf{BDesign}[\mathbf{Mat}(\mathbb{N}), \mathbf{FSet}]$ to an incidence matrix in $\mathbf{Block}$ by sending each matrix entry greater than 0 to 1. The uniformity, regularity and $\lambda$-balance conditions of the design construction ensure that the incidence matrix we obtain that way, represents a uniform, regular and $\lambda$-balanced design. On morphisms the functor acts as the identity.  
\end{proof}
Similarly, one can argue that the following holds.
\begin{theorem}
    There exists a functor $G:\mathbf{Design}[\mathbf{Mat}(\mathbb{N}), \mathbf{FSet}] \longrightarrow \mathbf{Design}$.
\end{theorem}
Note that this indicates that the categories $\mathbf{Design}[\mathbf{Mat}(\mathbb{N}), \mathbf{FSet}]$ and $\mathbf{BDesign}[\mathbf{Mat}(\mathbb{N}), \mathbf{FSet}]$ actually define a more general concept of (block)design. We will refer to it as \emph{categorical (block)designs}.
\begin{lemma}\label{lem. equivalence BDesignMAtN BDesign CPcFhilb}
   The category $\mathbf{BDesign}[\mathbf{Mat}(\mathbb{N}), \mathbf{FSet}]$ is equivalent to $\mathbf{BDesign}[\mathbf{CP_c}[\mathbf{FHilb}], \mathbf{FSet}]$.
\end{lemma}
\begin{proof}
    According to Proposition~\ref{prop. equivalence CPc FHilb and MatN} the categories $\mathbf{CP_c}[\mathbf{FHilb}]$ and $\mathbf{Mat}(\mathbb{N})$ are equivalent. Using Theorem \ref{thrm. equivalence arrow categories} from Appendix~\ref{app. structures in arrowcategories}, this gives rise to an equivalence between their arrow categories. 
\end{proof}

\subsection{The Category of Quantum Designs} \label{sec. results - qdesigns}

In this section we will define a category of quantum designs by applying the design construction to the category $\mathbf{CP[FHilb]}$. Moreover, we will show that this category contains two important subcategories: $\mathbf{QDesign}_{\mathrm B}$ which has objects that are uniform and regular quantum designs of degree 1, and $\mathbf{QDesign}_{\mathrm{RU}}$ that has uniform and regular quantum designs as objects. We will demonstrate that the latter actually contains a subcategory $\mathbf{MUB}$, with objects that are sets of mutually unbiased bases.

Recall from Section~\ref{sec. category theory} that $\mathbf{CP[FHilb]}$ is comprised of finite dimensional $H^{*}$-algebras and completely positive maps. Applying the design construction using the identity functor, we get a category $\mathbf{Design}[\mathbf{CP[FHilb]}]$, with objects that are CP\-maps between finite dimensional $H^{*}$-algebras, and morphisms that are pairs of CP\-maps.

\begin{definition}
The category $\mathbf{QDesign}$ is defined to be the category $\mathbf{Design}[\mathbf{CP}[\FHilb]]$.
\end{definition}
\begin{definition}
 The subcategory $\mathbf{RUDesign}[\mathbf{CP}[\mathbf{FHilb}]]$ of $\mathbf{QDesign}$ is called $\mathbf{QDesign}_{\mathrm{RU}}$. Its objects are uniform and regular quantum designs.   
\end{definition}
\begin{definition}
 The subcategory $\mathbf{BDesign}[\mathbf{CP}[\mathbf{FHilb}]]$ of $\mathbf{QDesign}_{\mathrm{RU}}$ is called $\mathbf{QDesign}_{\mathrm B}$. Its objects are uniform and regular quantum designs with degree 1.  
\end{definition}
\begin{example}\label{ex. fully quantum design}
Consider the subcategory of $\QDesign_\mathrm B$ where all objects are CP-maps between matrix algebras: $\phi: H \otimes H^* \to K\otimes K^{*}$ where $\mathrm{dim}(H)=b $ and $\mathrm{dim}(K)= v$. Because $H$ is a special Frobenius algebra, we get $\mathrm{dim}(H)= \mathrm{tr}(\mathrm{id}_{H})$. We then have the following uniformity and regularity conditions:
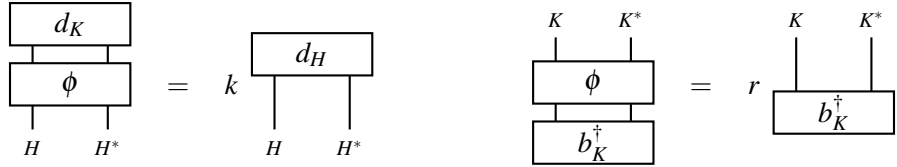
\begin{align*}
\begin{aligned}
\begin{tikzpicture}[thick, yscale=.8]
\draw (-.5,.25) node [below, label] {$H\vphantom{{}^*}$} to +(0,2);
\draw (.5,0.25) node [below, label] {$H^*$} to +(0,2);
\node [widthtwo] at (0,1) {$\phi$};
\node [widthtwo] at (0,2) {$d_K$};
\end{tikzpicture}
\end{aligned}
&\quad=\quad
k
\,\,\,\,
\begin{aligned}
\begin{tikzpicture}[thick, yscale=.8]
\draw (-.5,.25) node [below, label] {$H\vphantom{{}^*}$} to +(0,2);
\draw (.5,0.25) node [below, label] {$H^*$} to +(0,2);
\node [widthtwo, white] at (0,2) {$d_H$};
\node [widthtwo] at (0,1.5) {$d_H$};
\end{tikzpicture}
\end{aligned}
&
\begin{aligned}
\begin{tikzpicture}[thick, yscale=.8]
\draw (-.5,.75) to +(0,2) node [above, label] {$K\vphantom{{}^*}$};
\draw (.5,.75) to +(0,2) node [above, label] {$K^*$};
\node [widthtwo] at (0,2) {$\phi$};
\node [widthtwo] at (0,1) {$b^{\dagger}_{K}$};
\end{tikzpicture}
\end{aligned}
&\quad=\quad
r
\,\,\,\,
\begin{aligned}
\begin{tikzpicture}[thick, yscale=.8]
\draw (-.5,-.25) node [above, label] {$K\vphantom{{}^*}$} to +(0,-2);
\draw (.5,-0.25) node [above, label] {$K^*$} to +(0,-2);
\node [widthtwo, white] at (0,-2) {$d_K$};
\node [widthtwo] at (0,-1.5) {$b^{\dagger}_K$};
\end{tikzpicture}
\end{aligned}
\end{align*}
The $\lambda$-condition is given by:
$$
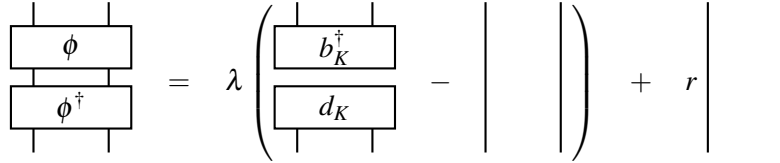
\begin{aligned}
\begin{tikzpicture}[thick, yscale=.8]
\draw (-.5,.25) to +(0,2.5);
\draw (.5,.25) to +(0,2.5);
\node [widthtwo] at (0,1) {$\phi^\dag$};
\node [widthtwo] at (0,2) {$\phi$};
\end{tikzpicture}
\end{aligned}
\quad=\quad
\lambda \left(
\begin{aligned}
\begin{tikzpicture}[thick, yscale=.8]
\draw (-.5,.25) to +(0,1);
\draw (.5,.25) to +(0,1);
\draw (-.5,1.75) to +(0,1);
\draw (.5,1.75) to +(0,1);
\node [widthtwo] at (0,1) {$d_K$};
\node [widthtwo] at (0,2) {$b^{\dagger}_K$};
\end{tikzpicture}
\end{aligned}
\quad-\quad
\begin{aligned}
\begin{tikzpicture}[thick, yscale=.8]
\draw (-.5,.25) to +(0,2.5);
\draw (.5,.25) to +(0,2.5);
\end{tikzpicture}
\end{aligned}
\,\,\,
\right)
\quad+\quad
r
\,\,\,
\begin{aligned}
\begin{tikzpicture}[thick, yscale=.8]
\draw (-.5,.25) to +(0,2.5);
\draw (.5,.25) to +(0,2.5);
\end{tikzpicture}
\end{aligned}
$$ 
Let $H=\mathbb{C}^2  = K$ and consider the CP-map $\varphi: \mathbb{C}^{2}  \otimes \mathbb{C}^{2} \to  \mathbb{C}^{2} \otimes \mathbb{C}^{2}$ with matrix representation:
$$
\begin{pmatrix}
        1 & 0 & 0& 1\\
        0 & \frac{1}{2} & \frac{1}{2} & 0\\
        0 & \frac{1}{2} & \frac{1}{2} & 0\\
        1 & 0 & 0& 1
    \end{pmatrix}
$$
This map represents a quantum design with parameters $\lambda=k=r=2=v=b$.
\end{example}
\begin{theorem}
    Every $1$-uniform, $r$-regular and $\lambda$-balanced quantum design of the form $S: H\otimes H^{*} \to K \otimes K^{*}$ defines a superoperator.
\end{theorem}
\begin{proof}
    By definition, $S$ is completely positive. Applying the uniformity condition to $S(\rho)$, where $\rho$ is a an arbitrary state in $H\otimes H^*$ shows that $S$ is also trace-preserving.
\end{proof}
In the previous example we have considered completely positive maps from a non-commutative algebra to a non-commutative algebra in $\mathbf{FHilb}$.
We can also consider CP-maps from a commutative algebra to a non-commutative algebra, i. e. maps of the form: $\varphi: H \to K^{*}\otimes K$.
In fact, we can encode uniform and regular quantum designs of degree 1 according to Zauner's notion via these maps.
\begin{theorem}\label{thrm Zauner cat}
There exists a subcategory of $\mathbf{QDesign}_\mathrm B$ that has objects that represent uniform and regular quantum designs of degree 1 according to Zauner's notion. 
\end{theorem}
\begin{proof}
Consider a uniform, regular and $\lambda$-balanced quantum design $\mathit{D}= \{ p_{1}, ..., p_{v} \}$, where each $p_{i}$ is a $b\times b$ projection matrix in a Hilbert space $\mathbb{C}^{b}$. Following Example \ref{ex. POVM CPM}, these projections $p_{i}: \mathbb{C}^{b} \to \mathbb{C}^{b}$ then give rise to a completely positive map $\phi: \mathbb{C}^{v} \to \mathbb{C}^{b} \otimes \mathbb{C}^{b}$ in $\mathbf{FHilb}$. This is valid because imposing uniformity, regularity and being $\lambda$-balanced on the projector has no impact on the CP-condition. Now take $\phi' = \phi^\dag : \mathbb{C}^b \otimes \mathbb{C}^b \to \mathbb{C}^v$, i. e. 
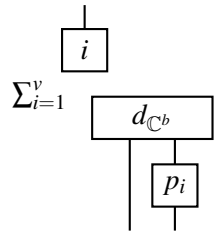
$$
\begin{tikzpicture}[scale=.6, thick]
\draw (1,4) to (1,5) node [above, label] {};
\draw (2,0) to (2,2.5); 
\draw (3,0)  to (3,2.5);
\node [widthone] at (3,1) {$p_{i}$};
\node [widthtwo] at (2.5,2.5) {$d_{\mathbb{C}^{b}}$};
\node [widthone] at (1,4) {$i$};
\node  at (0,3) {$\sum_{i=1}^{v}$};
\end{tikzpicture}
$$
Then we find:
\begin{align*}
\begin{aligned}
\begin{tikzpicture}[thick]
\node [dot] at (0,2) {};
\draw (0,1) to (0,2);
\draw (-.5,0) to +(0,1);
\draw (.5,0) to +(0,1);
\node [widthtwo] at (0,1) {$\phi'$};
\end{tikzpicture}
\end{aligned}
&\quad=\quad
k\,\,\,\,
\begin{aligned}
\begin{tikzpicture}[thick]
\node [dot, white] at (0,0) {};
\draw (-.5,0) to +(0,1);
\draw (.5,0) to +(0,1);
\node [widthtwo] at (0,1) {$d_{{\mathbb C}^b}$};
\end{tikzpicture}
\end{aligned}
&
\begin{aligned}
\begin{tikzpicture}[thick]
\draw (-.5,0) to +(0,1);
\draw (.5,0) to +(0,1);
\draw (0,1) to +(0,1);
\node [widthtwo] at (0,1) {$\phi'$};
\node [widthtwo] at (0,0) {$b_{\mathbb C ^b}$};
\end{tikzpicture}
\end{aligned}
\quad=\quad
r\,\,\,
\begin{aligned}
\begin{tikzpicture}[thick]
\draw (0,0) to (0,2);
\node [dot] at (0,0) {};
\end{tikzpicture}
\end{aligned}
\end{align*}
$$
\begin{aligned}
\begin{tikzpicture}[thick]
\draw (0,0.25) to +(0,1);
\draw (0,2.75) to +(0,-1);
\draw (-.5,1) to +(0,1);
\draw (.5,1) to +(0,1);
\node [widthtwo] at (0,1) {$\phi^{\dag \prime}$};
\node [widthtwo] at (0,2) {$\phi'$};
\end{tikzpicture}
\end{aligned}
\quad=\quad
\lambda \left(
\,\,\,
\begin{aligned}
\begin{tikzpicture}[thick]
\draw (0,0.25) to +(0,.75);
\draw (0,2.75) to +(0,-.75);
\node [dot] at (0,1) {};
\node [dot] at (0,2) {};
\end{tikzpicture}
\end{aligned}
\quad-\quad
\begin{aligned}
\begin{tikzpicture}[thick]
\draw (0,0.25) to +(0,2.5);
\end{tikzpicture}
\end{aligned}
\,\,\,\,
\right)
\quad+\quad
r
\,\,\,
\begin{aligned}
\begin{tikzpicture}[thick]
\draw (0,0.25) to +(0,2.5);
\end{tikzpicture}
\end{aligned}
$$
These coincide with the conditions given by the design construction applied to $\mathbf{FHilb}$. Moreover, we can recover Eq.~\ref{eq. 1 qdesign} and Eq.~\ref{eq. 2 qdesign} from Lemma~\ref{lem. properties of quantum designs} as $\mathrm{dim}(\mathbb{C}^{v})= v$ and $\mathrm{dim}(\mathbb{C}^{b})= b$ and Eq.~\ref{eq. bd1 cat} and Eq.~\ref{eq. bd2 cat} hold.
\end{proof}
\begin{theorem}
There exists a subcategory $\mathbf{MUB}$ of $\mathbf{QDesign}_{\mathrm{RU}}$ with  objects that are collections of MUBs and morphisms that are pairs of functions.
\end{theorem}
\begin{proof}
Consider the CP-map $M : \mathbb{C}^{k\cdot d} \cong \mathbb{C}^{d}\otimes \mathbb{C}^{k} \to \mathbb{C}^{d} \otimes \mathbb{C}^{d}$:
$$
M \quad=\quad
\sum_{a=1}^k \sum_{i=1}^d \quad
\begin{aligned}
\begin{tikzpicture}[thick]
\draw  (0,2) to +(0,1) ;
\node [widthone] at (0,2) {$i$};
\draw  (1,2) to +(0,1) ;
\node [widthone] at (1, 2) {$a$};
\draw (1,0) to +(0,1.25);
\draw (0,0) to +(0,1.25);
\node [widthtwo] at (0.5,1.25) {$d_{{\mathbb C} ^d}$};
\node [widthone] at (1,0.5) {$p_i^a$};
\end{tikzpicture}
\end{aligned}
$$
This map satisfies the following equations:
\begin{align*}
\begin{aligned}
\begin{tikzpicture}[thick, yscale=.8]
\node [dot] at (0,2) {};
\node [dot] at (1,2) {};
\draw (0,0) to +(0,2);
\draw (1,0) to +(0,2);
\node [widthtwo] at (.5,1) {$M$};
\end{tikzpicture}
\end{aligned}
&\quad=\quad k \,\,\,\,
\begin{aligned}
\begin{tikzpicture}[thick]
\draw (0,0) to +(0,1);
\draw (1,0) to +(0,1);
\node [widthtwo] at (.5,1) {$d_{\mathbb C ^d}$};
\node [dot, white] at (1,0) {};
\end{tikzpicture}
\end{aligned}
&
\begin{aligned}
\begin{tikzpicture}[thick, yscale=.8]
\draw (0,0) to +(0,2);
\draw (1,0) to +(0,2);
\node [widthtwo] at (.5,1) {$M$};
\node [widthtwo] at (.5,0) {$b_{\mathbb C ^d}$};
\end{tikzpicture}
\end{aligned}
&\quad=\quad
\begin{aligned}
\begin{tikzpicture}[thick]
\draw (0,0) to +(0,2);
\draw (1,0) to +(0,2);
\node [dot] at (1,0) {};
\node [dot] at (0,0) {};
\end{tikzpicture}
\end{aligned}
\end{align*}
$$
\begin{aligned}
\begin{tikzpicture}[thick, yscale=.8]
\draw (-.5,.25) to +(0,2.5);
\draw (.5,.25) to +(0,2.5);
\node [widthtwo] at (0,1) {$M^\dag$};
\node [widthtwo] at (0,2) {$M$};
\end{tikzpicture}
\end{aligned}
\quad=\quad
\frac 1 d  \left(\,\,
\begin{aligned}
\begin{tikzpicture}[thick, xscale=.7, yscale=.8]
\draw (-.5,.25) to +(0,.75);
\draw (.5,.25) to +(0,.75);
\draw (-.5,2) to +(0,.75);
\draw (.5,2) to +(0,.75);
\node [dot] at (-.5,1) {};
\node [dot] at (.5,1) {};
\node [dot] at (-.5,2) {};
\node [dot] at (.5,2) {};
\end{tikzpicture}
\end{aligned}
\quad-\quad
\begin{aligned}
\begin{tikzpicture}[thick, xscale=.7, yscale=.8]
\draw (-.5,.25) to +(0,.75);
\draw (-.5,2) to +(0,.75);
\draw (.5,.25) to +(0,2.5);
\node [dot] at (-.5,1) {};
\node [dot] at (-.5,2) {};
\end{tikzpicture}
\end{aligned}
\,\,\,
\right)
\quad+\quad
\begin{aligned}
\begin{tikzpicture}[thick, xscale=.7, yscale=.8]
\draw (-.5,.25) to +(0,2.5);
\draw (.5,.25) to +(0,2.5);
\end{tikzpicture}
\end{aligned}
$$
Here the last equation can be understood as a generalised $\lambda$-equation. Restricting $\mathbf{QDesign}_{RU}$ to objects of this form, we get a category that has objects that are $1$-uniform and $k$-regular quantum designs of degree 2 where $\Lambda = \{ \frac{1}{d}, 0 \}$, i. e. $k$ MUBs in dimension $d$.
\end{proof}
\subsection{Relating \textbf{BDesign} to $\mathbf{QDesign}$} \label{sec. results - functor}
In this section we will construct a functor between $\mathbf{BDesign}[\mathbf{Mat}(\mathbb{N}), \mathbf{FSet}]$  and $\mathbf{QDesign}_{B}$. 
\begin{satz}\label{satz: existence functor between 2design and qdesign}
There exists a functor $Q: \mathbf{BDesign}[\mathbf{Mat}(\mathbb{N}), \mathbf{FSet}] \to \mathbf{QDesign}_{B}$ that relates a generalized balanced incomplete block designs to uniform, regular and $\lambda$-balanced quantum designs.
\end{satz}
\begin{proof}
According to Lemma~\ref{lem. equivalence BDesignMAtN BDesign CPcFhilb} the categories $\mathbf{BDesign}[\mathbf{Mat}(\mathbb{N}), \mathbf{FSet}]$ and $\mathbf{BDesign}[\mathbf{CP}_{c}[\mathbf{FHilb}], \mathbf{FSet}]$ are equivalent. So we can actually represent an arbitrary object $\chi: b\to v$ in $\mathbf{BDesign}[\mathbf{Mat}(\mathbb{N}), \mathbf{FSet}]$ via a uniform, regular and $\lambda$-balanced CP-map  $\chi: \mathbb{C}^{b} \to \mathbb{C}^{v}$. The functor $Q$ acts on objects by sending each object $\chi: \mathbb{C}^{b} \to \mathbb{C}^{v}$ in $\mathbf{BDesign}[\mathbf{Mat}(\mathbb{N}), \mathbf{FSet}]$ with parameters $k$, $r$ and $\lambda$ to the map
$\phi=  \chi \circ L : \mathbb{C}^{b} \otimes\mathbb{C}^{b} \to \mathbb{C}^{b} \to \mathbb{C}^{v}$, where the map $L: \mathbb{C}^{b} \otimes \mathbb{C}^{b} \to \mathbb{C}^{b}$ is the so-called Cayley embedding, which in our case simply becomes the multiplication $\mu: \mathbb{C}^{b} \otimes \mathbb{C}^{b} \to \mathbb{C}^{b}$ as we have that $A=\mathbb{C}^{b}\cong (\mathbb{C}^{b})^{*}= A^{*}$. Its conjugate $L^{\dagger}$ is just the comultiplication $\Delta: \mathbb{C}^{b} \to \mathbb{C}^{b} \otimes \mathbb{C}^{b}$.
The resulting map $\phi$ is as concatenation of completely positive maps also completely positive. We  depict this via the following string diagram:
$$
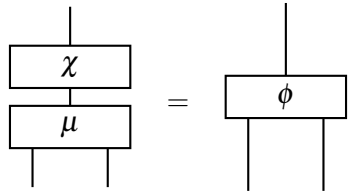
\begin{aligned}
\begin{tikzpicture}[thick, yscale=.8]
\draw (0,1) to +(0,2);
\draw (-.5,0) to (-.5,1);
\draw (.5,0) to (.5,1);
\node [widthtwo] at (0,1) {$\mu$};
\node [widthtwo] at (0,2) {$\chi$};
\end{tikzpicture}
\end{aligned}
\quad=\quad
\begin{aligned}
\begin{tikzpicture}[thick]
\draw (0,0) to +(0,1.25);
\draw (-.5,0) to +(0,-1.25);
\draw (.5,0) to +(0,-1.25);
\node [widthtwo] at (0,0) {$\phi$};
\end{tikzpicture}
\end{aligned}
$$
Via concatenation, each morphism in $\mathbf{BDesign}[\mathbf{Mat}(\mathbb{N}), \mathbf{FSet}]$
\[\begin{tikzcd}
\mathbb{C}^{b}\arrow{d}{\chi} \arrow{r}{\xi'} & \mathbb{C}^{b'} \arrow{d}{\chi'}\\
\mathbb{C}^{v} \arrow{r}{\xi} & \mathbb{C}^{v'}
\end{tikzcd}
\]
gets mapped to a morphism in $\mathbf{QDesign}_{B}$, as follows:
\[\begin{tikzcd}
\mathbb{C}^{b}\otimes \mathbb{C}^{b} \arrow{d}{\mu} \arrow{r}{\xi'\otimes \xi '} & \mathbb{C}^{b'} \otimes\mathbb{C}^{b'}\arrow{d}{\mu}\\
\mathbb{C}^{b}\arrow{d}{\chi}\arrow{r}{\xi'} &  \mathbb{C}^{b'}\arrow{d}{\chi'}\\
\mathbb{C}^{v} \arrow{r}{\xi} & \mathbb{C}^{v'} 
\end{tikzcd}
\] 
This diagram commutes, because $\xi'$ can be extended to a morphism of monoids as $\xi'$ is a function. It is easy to verify that this functor respects composition and sends the identity morphism in 
$\mathbf{BDesign}[\mathbf{Mat}(\mathbb{N}), \mathbf{FSet}]$, i. e. $\mathrm{id}_{\psi}=(\mathrm{id}, \mathrm{id})$, to the identity morphism $\mathrm{id}_{Q(\psi)}=(\mathrm{id}\otimes \mathrm{id}, \mathrm{id}\otimes \mathrm{id})$ in $\mathbf{QDesign}_{\mathrm B}$. 
The regularity condition then becomes:
$$
k\,\,\,
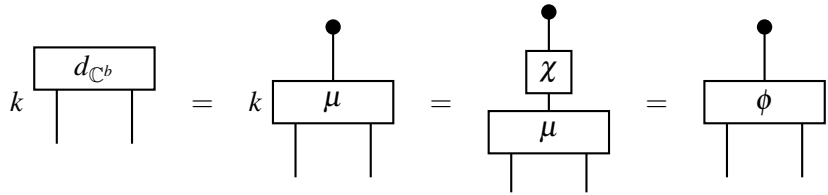
\begin{aligned}
\begin{tikzpicture}[thick]
\draw (-.5,-1) to +(0,1);
\draw (.5,-1) to +(0,1);
\node [widthtwo] at (0,0) {$d_{\mathbb C ^b}$};
\node [dot,white] at (0,-1) {};
\end{tikzpicture}
\end{aligned}
\quad=\quad
k\,\,\,
\begin{aligned}
\begin{tikzpicture}[thick]
\draw (-.5,-1) to +(0,1);
\draw (.5,-1) to +(0,1);
\draw (0,0) to +(0,1);
\node [widthtwo] at (0,0) {$\mu$};
\node [dot] at (0,1) {};
\end{tikzpicture}
\end{aligned}
\quad=\quad
\begin{aligned}
\begin{tikzpicture}[thick, yscale=.8]
\draw (-.5,-1) to +(0,1);
\draw (.5,-1) to +(0,1);
\draw (0,1) to +(0,1);
\draw (0,0) to +(0,1);
\node [widthone] at (0,1) {$\chi$};
\node [widthtwo] at (0,0) {$\mu$};
\node [dot] at (0,2) {};
\end{tikzpicture}
\end{aligned}
\quad=\quad
\begin{aligned}
\begin{tikzpicture}[thick]
\draw (-.5,-1) to +(0,1);
\draw (.5,-1) to +(0,1);
\draw (0,0) to +(0,1);
\node [widthtwo] at (0,0) {$\phi$};
\node [dot] at (0,1) {};
\end{tikzpicture}
\end{aligned}
$$
which is exactly the regularity condition in $\mathbf{QDesign}_{\mathrm B}$. Note that we have used the fact that $\mathbb{C}^{b}$ is a special Frobenius algebra in the first step.
in the second step. For uniformity we find:
$$
r\,\,\,
\begin{aligned}
\begin{tikzpicture}[thick]
\draw (0,-1) to (0,1);
\node [dot] at (0,-1) {};
\end{tikzpicture}
\end{aligned}
\quad=\quad
\begin{aligned}
\begin{tikzpicture}[thick, yscale=.8]
\draw (0,-1) to +(0,1);
\draw (-.5,0) to +(0,1);
\draw (.5,0) to +(0,1);
\draw (0,1) to +(0,2);
\node [widthone] at (0,2) {$\chi$};
\node [widthtwo] at (0,0) {$\Delta$};
\node [widthtwo] at (0,1) {$\mu$};
\node [dot] at (0,-1) {};
\end{tikzpicture}
\end{aligned}
\quad=\quad
\begin{aligned}
\begin{tikzpicture}[thick]
\draw (0,0) to (0,1);
\draw (-.5,1) to +(0,1);
\draw (.5,1) to +(0,1);
\draw (0,2) to +(0,1);
\node [widthtwo] at (0,2) {$\phi$};
\node [widthtwo] at (0,1) {$\Delta$};
\node [dot] at (0,0) {};
\end{tikzpicture}
\end{aligned}
\quad=\quad
\begin{aligned}
\begin{tikzpicture}[thick]
\draw (0,2) to (0,3);
\draw (-.5,1) to +(0,1);
\draw (.5,1) to +(0,1);
\node [widthtwo] at (0,2) {$\phi$};
\node [widthtwo] at (0,1) {$b_{\mathbb C ^b}$};
\end{tikzpicture}
\end{aligned}
$$
which is precisely the uniformity condition in $\mathbf{QDesign}_{\mathrm B}$.
In a similar way one can verify that the $\lambda$\-condition in $\mathbf{BDesign}[\mathbf{Mat}(\mathbb{N}), \mathbf{FSet}]$ gets mapped to the $\lambda$-condition in $\mathbf{QDesign}_{\mathrm B}$.
\end{proof}
In this construction every classical design gives rise to a uniform and regular quantum design of degree 1, analogously to Theorem \ref{thrm equivalence designs}. However, it is straightforward to verify that the functor $Q$ does not yield an equivalence of categories, as it is not essentially surjective.

A widely-discussed topic is the existence of MUBs in non-primepower dimensions. One can ask if it is possible to extend the functor $Q$ to a functor $\widetilde{Q}: \mathbf{BDesign}[\mathbf{Mat}(\mathbb{N}), \mathbf{FSet}] \to \mathbf{MUB}$ that maps a classical design to a set of MUBs. We conjecture that there does not exist a classical design that gets sent to a $\mathbf{MUB}$ via $Q$. 

\appendix

\section{Arrow categories} \label{app. structures in arrowcategories}
In this section we will review the concept of arrow categories and derive some facts about arrow categories which we believe to be new. We will focus only on content relevant for the overall purpose of this paper. However, there are more results on that topic, covering Hopf algebras, Frobenius structures and topological field theories in arrow categories which can be found in~\cite{goedicke2023structures}.

\begin{definition}[See \cite{IntroCat}, pages 23-24]\label{def. arrow category}
For a category \C, its arrow category $\mathbf{Arr}[\C]$ is defined as follows:
\begin{itemize}
    \item objects are triples $(A, B, h)$ with $h: A\to B$ in \C;
    \item morphisms $\phi: (A, B, h) \to (A', B', h')$ are pairs of morphisms $\phi_{A}: A \to A'$ and $\phi_{B}: B \to B'$ in $\mathcal{C}$ such that the following diagram commutes in \C:
    \[\begin{tikzcd}
    A \arrow[swap]{d}{h} \arrow{r}{\phi_A} &     \arrow{d}{h'} A' \\
    B \arrow{r}{\phi_B} & B'
\end{tikzcd}
\]
\end{itemize}
\end{definition}

In the following we will show that an arrow category inherits certain structures from their underlying category. This includes functors, natural transformations and the monoidal product.
\begin{satz}\label{prop. functor between arrow cat.}
 Given a functor $F: \mathcal{C} \to \mathcal{D}$, we apply the arrow construction to obtain a functor \mbox{$\widetilde F: \mathbf{Arr}[\C] \to \mathbf{Arr}[\mathcal D]$.}
\end{satz}
\begin{proof}
Given a functor $F: \mathcal{C} \to \mathcal{D}$ , we can define a functor $\widetilde F : \mathbf{Arr}[\C] \to \mathbf{Arr}[\mathcal D]$ as follows. On objects, we map $f : A\to B$ in $\mathbf{Arr}[\C]$ to an object $F(f): F(A) \to F(B)$ in $\mathbf{Arr}[\mathcal{D}]$. On morphisms, we map  $(\phi, \psi): f \to f'$ in $\mathbf{Arr}[\C]$ to a morphism $\tilde{F}(\phi, \psi)=(F(\phi), F(\psi)): F(f) \to F(f')$ in $\mathbf{Arr}[\mathcal{D}]$. This is valid because the diagram
\[\begin{tikzcd}
    F(A) \arrow{r}{F(\phi)} \arrow[swap]{d}{F(f)} &  \arrow{d}{ F(f')} F(A') \\
    F(B) \arrow{r}{F(\psi)} & F(B')
\end{tikzcd}
\]
commutes due to functoriality of $F$. Moreover, we have
\begin{align}
    \widetilde{F}(\mathrm{id}_{A}, \mathrm{id}_{B}) = (F(\mathrm{id}_{A}), F(\mathrm{id}_{B})) =(\mathrm{id}_{F(A)}, \mathrm{id}_{F(B)}),
\end{align} where $(\mathrm{id}_{A}, \mathrm{id}_{B})$ is the identity morphism in $\mathbf{Arr}[\C]$. Due to functoriality of $F$ and because the concatenation of two commuting diagrams yields again a commuting diagram, $\tilde{F}$ also preserves composition.
\end{proof}\noindent
Similarly, a contravariant functor $F: \mathcal{C} \to \mathcal{D}$ gives rise to a contravariant functor $\widetilde{F}: \mathbf{Arr}[\mathcal{C}] \to \mathbf{Arr}[\mathcal{D}]$.
\begin{satz}\label{prop. natural trafo between arrow cats}
Let $F,G:\mathcal C \to \mathcal D$ be two functors between two categories $\C$ and $\mathcal D$, and let $\widetilde{F}, \widetilde{G}: \mathbf{Arr}[\C] \to \mathbf{Arr}[\mathcal D]$ be the induced functors on the arrow categories. A natural transformation \mbox{$\eta: F \Rightarrow G$} induces a natural transformation $\tilde \eta: \tilde F \Rightarrow \tilde G$.
\end{satz} 
\begin{proof}
Let $\eta: F \Rightarrow G$ be a natural transformation that assigns to every object $A$ in $\mathcal{C}$ a morphism $\eta_{A}: F(A) \to G(A)$, such that for any morphism $f:A \to B$ in $\mathcal{C}$ the following diagram (naturality condition) commutes:
\[\begin{tikzcd}
    F(A) \arrow{r}{\eta_{A}} \arrow[swap]{d}{F(f)} &  \arrow{d}{ G(f)} G(A) \\
    F(B) \arrow{r}{\eta_{B}} & G(B)
\end{tikzcd}
\]
One can use the naturality of $\eta$ to define a natural transformation $\widetilde{\eta}: \widetilde{F} \Rightarrow \widetilde{G} $ that assigns to every object $f: A \to B$ in $\mathbf{Arr}[\C]$ a morphism $\widetilde{\eta}_{f}=(\eta_{A}, \eta_{B}): \widetilde{F}(f) \to \widetilde{G}(f)$ via the commutative diagram from above, such that for any morphism $(\phi, \psi): f \to f'$ in $\mathbf{Arr}[\C]$: 
\[\begin{tikzcd}
   A \arrow{r}{\phi} \arrow[swap]{d}{f} &  \arrow{d}{ f'} A' \\
   B \arrow{r}{\psi} & B'
\end{tikzcd}
\]
the following diagram (naturality condition in the arrow category) commutes:
\[\begin{tikzcd}[row sep=1cm, column sep=1cm, inner sep=2ex]
&F(A) \arrow[swap]{dl}{F(\phi)}\arrow{rr}{\eta_{A}} \arrow[dashed]{dd}[yshift=-0.5cm]{F(f)} & & G(A) \arrow[swap]{dl}{G(\phi)} \arrow{dd}{G(f)} \\
F(A') \arrow{rr}[xshift= 0.5cm]{\eta_{A'}} \arrow{dd} {F(f')} & & G(A') \arrow[]{dd} [yshift=0.5cm]{G(f')}\\
&F(B) \arrow[dashed]{dl}{F(\psi)} \arrow[dashed] {rr}[xshift= -0.5cm]{\eta_{B}} & & G(B) \arrow{dl}{G (\psi)} \\
F(B') \arrow[swap]{rr}{\eta_{B'}} & & G (B')\\
\end{tikzcd}
\]
Here the the top, the back, the front and the bottom face commute due to naturality of $\eta$ and the two side faces commute by definition. Hence the whole diagram commutes and we have defined a natural transformation $\widetilde{\eta}: \widetilde{F} \Rightarrow \widetilde{G}$.
\end{proof}\noindent
\begin{satz}\label{rem. natural isomoprhism between arrow cats}
  If $\eta: F \Rightarrow G$ is a natural isomorphism, then so is $\Tilde{\eta}: \widetilde{F} \Rightarrow \widetilde{G}$.  
\end{satz}

\begin{theorem}\label{thrm. equivalence arrow categories}
Let $\mathcal{C}$ and $\mathcal{D}$ be 
equivalent categories; that is, there exist  functors $F: \mathcal{C} \to \mathcal{D}$ and $G: \mathcal{D} \to \mathcal{C}$ and natural isomorphisms $F \circ G \cong \mathrm{id}_{\mathcal{D}}$ and $G \circ F \cong \mathrm{id}_{\mathcal{C}}$. Then Arr($\mathcal{C}$) and Arr($\mathcal{D}$) are also equivalent.
\end{theorem}
\begin{proof}
By Proposition~\ref{prop. functor between arrow cat.} the functors $F: \mathcal{C} \to \mathcal{D}$ and $G: \mathcal{D} \to \mathcal{C}$ give rise to functors $\widetilde{F}: \mathbf{Arr}[\mathcal{C}] \to \mathbf{Arr}[\mathcal{D}]$ and $\widetilde{G}: \mathbf{Arr}[\mathcal{D}] \to \mathbf{Arr}[\mathcal{C}]$. From Proposition~\ref{rem. natural isomoprhism between arrow cats} we know that the natural isomorphisms $F \circ G \cong \mathrm{id}_{\mathcal{D}}$ and $G \circ F \cong \mathrm{id}_{\mathcal{C}}$ give rise to natural isomorphisms $\widetilde{F} \circ \widetilde{G} \cong \mathrm{id}_{\mathbf{Arr}[\mathcal{D}]}$ and $\widetilde{G} \circ \widetilde{F} \cong \mathrm{id}_{\mathbf{Arr}[\mathcal{C}]}$. Hence we have an equivalence.
\end{proof}
One can show that the same theorems apply to monoidal functors and monoidal natural transformations~\cite{goedicke2023structures}. Moreover, one can define a monoidal product in the arrow category of a monoidal category as the following proposition will show. However, this result is not necessarily new and can be in fact found in a similar notion in~\cite{white2018arrow}. 
\begin{satz}\label{prop. monoidal product in arrow cats}
For a monoidal category \C, we can define a monoidal product on $\mathbf{Arr}[\C]$, written $\boxtimes$, as follows:
\begin{itemize}
\item on objects, $f \boxtimes g := f \otimes g$;
\item on morphisms,  $(p, q) \boxtimes (p', q') := (p\otimes p', q\otimes q')$.
\end{itemize}
\end{satz}
\begin{proof}
We will show that the pentagon and the triangle axiom are satisfied. The pentagon axiom holds due to the following diagram, where the front and the back face commute because $\alpha$ satisfies the ordinary pentagon axiom. The two side faces commute due to the definition of the monoidal product and naturality of the associator, and the top and bottom faces commute due to naturality of the associator:
\[
\def\stimes{{\otimes}}
\hspace{-3cm}
\small
\begin{tikzcd}[ampersand replacement=\&]
\& A_{1} \stimes (A_{2} \stimes (A_{3} \stimes A_{4})) \arrow[swap]{dl}{f_{1} \stimes ( f_{2} \stimes (f_{3} \stimes f_{4}))}\arrow{r}{\alpha} \arrow[dashed]{dd}[yshift=0.75cm]{\mathrm{id}_{A_{1}} \stimes \alpha} \&  (A_{1}\stimes A_{2}) \stimes (A_{3} \stimes A_{4}) \arrow{dl}{(f_{1} \stimes f_{2}) \stimes (f_{3} \stimes f_{4})} \arrow{r}{\alpha}  \&  ((A_{1}\stimes A_{2}) \stimes A_{3}) \stimes A_{4} \arrow{dl}{((f_{1} \stimes f_{2}) \stimes f_{3} )\stimes f_{4}} \arrow{dd}{\alpha \stimes \mathrm{id}_{A_{4}}}
\\
B_{1}\stimes (B_{2} \stimes (B_{3} \stimes B_{4})) \arrow {r}{\alpha} \arrow{dd} {\mathrm{id}_{B_{1}} \stimes \alpha} \& |[]|(B_{1}\stimes B_{2}) \stimes (B_{3} \stimes B_{4})\arrow{r}{\alpha} \&   ((B_{1}\stimes B_{2}) \stimes B_{3}) \stimes B_{4} \arrow[]{dd}[yshift=1cm]{\alpha}\\
\& A_{1}\stimes ((A_{2} \stimes A_{3}) \stimes A_{4})\arrow[dashed]{dl}{f_{1} \stimes (f_{2}\stimes f_{3} )\stimes f_{4})} \arrow[dashed]{rr}[xshift=1cm]{\alpha} \& \& (A_{1}\stimes (A_{2} \stimes A_{3})) \stimes A_{4} \arrow{dl}{(f_{1} \stimes (f_{2} \stimes f_{3} ))\stimes f_{4}} \\
B_{1}\stimes ((B_{2} \stimes B_{3}) \stimes B_{4}) \arrow[swap]{rr}{\alpha} \& \& (B_{1}\stimes (B_{2} \stimes B_{3})) \stimes B_{4}\\
\end{tikzcd}
\hspace{-3cm}
\]
The triangle axiom for $\mathbf{Arr}[\C]$ is given by the following diagram:
 
\[\begin{tikzcd}[row sep=scriptsize, column sep=scriptsize]
(A\otimes \mathbb{I}) \otimes A' \arrow{dd}{(f\otimes \mathrm{id}_{\mathbb{I}}) \otimes f'} \arrow{dr}{\rho \otimes \mathrm{id}_{A'} } \arrow{rr}{\alpha} & &  A\otimes (\mathbb{I} \otimes A') \arrow{dd}{f\otimes (\mathrm{id}_{\mathbb{I}} \otimes f')} \arrow{dl}{\mathrm{id}_{A}\otimes \lambda } \\
&   A \otimes A \arrow{dd}[yshift = 0.5cm]{f\otimes f'}& \\
(B\otimes \mathbb{I}) \otimes B'\arrow{dr}{\rho \otimes \mathrm{id}_{B'} } \arrow[dashed]{rr}[xshift= 1cm]{\alpha} & & B\otimes (\mathbb{I} \otimes B') \arrow{dl}{\mathrm{id}_{B}\otimes \lambda} \\
 & B \otimes B'  & 
\end{tikzcd}
\hspace{-3cm}
\]
Here the top and the bottom faces commute due to the the triangle identity and the two side faces commute due to the definition of the monoidal product in $\mathbf{Arr}[\C]$ and due to naturality of the left and right unitors in $\mathcal{C}$. Finally, the back face commutes because of the naturality of the associator.
\end{proof}

\section{Proof for Lemma~\ref{lem. properties of quantum designs}} \label{app: proofs}
\begin{proof} 
Consider an $r$-uniform, $k$-regular and $\lambda$-balanced quantum design $D=\{ p_{1},...,p_{v} \}$ with $p_{i} \in \mathbb{C}^{b}$. By applying the trace function to the regularity condition, we get:
\begin{align}
    \mathrm{Tr}\Bigl(\sum_{i=0}^{v} p_i \Bigr) = \sum_{i=0}^{v} \mathrm{Tr}( p_i) = k \cdot \mathrm{Tr}(\mathbb{I}_{b \times b}).
\end{align}
Using the uniformity condition this expression becomes:
\begin{align}
    \sum_{i=0}^{v} \mathrm{Tr}( p_i) = \sum_{i=0}^{v} r = v \cdot r = k \cdot \mathrm{Tr}(\mathbb{I}_{b \times b}) = k \cdot b.
\end{align}
This proves Eq.~\ref{eq. 1 qdesign}.

In order to prove Eq.~\ref{eq. 2 qdesign}, we start with the following expression:
\begin{align}
    b =  \mathrm{Tr}(\mathbb{I}_{b \times b}) =  \mathrm{Tr}(\mathbb{I}_{b \times b}^{2}).
\end{align}
Using the regularity condition, we get:
\begin{align}
    \mathrm{Tr}(\mathbb{I}_{b \times b}^{2}) = \frac{1}{k^{2}}\mathrm{Tr} \Bigl(\sum_{i=0}^{v} p_i \sum_{j=0}^{v} p_i \Bigr) = \frac{1}{k^{2}} \sum_{i,j=0}^{v} \mathrm{Tr}(p_i p_j) = \frac{1}{k^{2}} \sum_{i,j=0, j\neq i}^{v} \lambda  + \frac{1}{k^{2}} \sum_{i}^{v} r = \frac{1}{k^{2}} (\lambda v(v-1) + vr)
\end{align}
Here we have used the uniformity and the $\lambda$-condition in the third step.
Hence we have:
\begin{align}
    b & =  \frac{1}{k^{2}} (\lambda v(v-1) + vr) \Leftrightarrow \\
    \frac{b \cdot k}{v} \cdot k & = \lambda (v-1) + r
\end{align}
Using Eq.~\ref{eq. 1 qdesign}, we obtain:
\begin{align}
    r \cdot k = \lambda (v-1) + r
\end{align}
This is equivalent to Eq.~\ref{eq. 2 qdesign}.
\end{proof}

\section{Structures in \textbf{RUDesign} and \textbf{QDesign}} \label{sec. structures in BDesign and QDesign}
In the following we will discuss some structures of the categories \textbf{RUDesign} and \textbf{QDesign}. 
\begin{theorem}\label{prop. RUDesign is monoidal category}
The category $\mathbf{RUDesign}[\mathbf{Mat}(\mathbb{N}), \mathbf{FSet}]$ is a monoidal category with monoidal product given by the Kronecker product between matrices. In particular, for objects $\chi:b \longrightarrow v$ and $\chi': b' \longrightarrow v'$ with parameters $k$ resp. $k'$ and $r$ resp. $r'$ we have that $\chi \otimes \chi': b\otimes b' \longrightarrow v \otimes v'$ is an object in $\mathbf{RUDesign}[\mathbf{Mat}(\mathbb{N}), \mathbf{FSet}]$ with parameters $k\cdot k'$ and $r\cdot r'$. 
\end{theorem}
\begin{proof}
Since $\mathbf{Mat}(\mathbb{N})$ is a monoidal category, we can apply Prop.~\ref{prop. monoidal product in arrow cats} from Appendix~\ref{app. structures in arrowcategories} to get a monoidal product on $\mathbf{Arr}[\mathbf{Mat}(\mathbb{N})]$. If we now restrict to matrices representing uniform and regular designs, every $\chi \otimes \chi': b\cdot b' \longrightarrow v \cdot v'$, where  $\chi:b \longrightarrow v$ and $\chi': b' \longrightarrow v'$ are objects in $\mathbf{Arr}[\mathbf{Mat}(\mathbb{N})]$ with parameters $k$ resp. $k'$ and $r$ resp. $r'$, fulfils the  uniformity and regularity conditions with parameters $k\cdot k'$ and $r\cdot r'$.
\end{proof}
Similarly, one can argue that the following can be derived from Prop.~\ref{prop. monoidal product in arrow cats}:
\begin{theorem}
    The category $\mathbf{Design}[\mathbf{Mat}(\mathbb{N}), \mathbf{FSet}]$ is a monoidal category, with monoidal product given by the Kronecker product between matrices.
\end{theorem}
As the Kronecker product of two matrices not necessarily fulfils the $\lambda$-condition, the category $\mathbf{BDesign}$ is not monoidal in general.

\begin{satz}
There exists a functor $\widetilde{D}: \mathbf{RUDesign}[\mathbf{Mat}(\mathbb{N}), \mathbf{FSet}] \longrightarrow \mathbf{RUDesign}[\mathbf{Mat}(\mathbb{N}), \mathbf{FSet}]$ that maps each $k$-uniform and $r$-regular design $\chi$ to its dual $\chi^{T}$ which is a $r$-uniform and $k$-regular design and each pair of functions $(N_{V}, N_{B}): \chi \longrightarrow \chi'$ to its transpose $(N_{V}^{T}, N_{B}^{T}): \chi' \longrightarrow \chi$.
\end{satz}
\begin{proof}
We can define a functor $D: \mathbf{Mat}(\mathbb{N}) \longrightarrow \mathbf{Mat}(\mathbb{N})$ that sends each natural number to itself and each matrix to its transpose. By remark \ref{prop. functor between arrow cat.} this functor then gives rise to a contravariant functor $\Tilde{D}: \mathrm{Arr}[\mathbf{Mat}(\mathbb{N})] \longrightarrow \mathrm{Arr}[\mathbf{Mat}(\mathbb{N})]$ that maps each object, i.e. a matrix, to its transpose and each pair of morphisms, i. e. a pair of matrices, to its transpose. If we now restrict to the subcategory where each object represents a uniform and regular design and all morphisms are pairs of functions, we get a functor: $\widetilde{D}: \mathbf{RUDesign}[\mathbf{Mat}(\mathbb{N}), \mathbf{FSet}] \longrightarrow \mathbf{RUDesign}[\mathbf{Mat}(\mathbb{N}), \mathbf{FSet}]$.
\end{proof}

Just as in the classical case, $\mathbf{QDesign}$ is also equipped with some structure.
\begin{theorem}
The category $\mathbf{QDesign}_{\mathrm{RU}}$ is a monoidal category.
\end{theorem}
\begin{proof}
The category $\mathbf{CP}[\mathbf{FHilb}]$ is monoidal~\cite{jamie}. According to Prop.~\ref{prop. monoidal product in arrow cats} this gives rise to a monoidal product in $\mathbf{Arr[CP[FHilb]]}$.
If we now restrict to the case where the objects in $\mathbf{Arr[CP[FHilb]]}$ encode uniform and regular quantum designs and the morphisms are pairs of functions, i. e. to the category $\mathbf{QDesign}_{RU}$, it is straightforward to verify that the tensor product of two uniform and regular CP-maps with parameters $k, r$ and $k^{\prime}, r^{\prime}$ respectively, again fulfils the uniformity and regularity condition with parameters $k\cdot k'$ and $r\cdot r'$. 
\end{proof}
Similarly, one can argue that the following has to hold:
\begin{theorem}
    The category $\mathbf{QDesign}$ is a monoidal category.
\end{theorem}
The monoidal product of two CP-maps satisfying the $\lambda$-condition does not satisfy the $\lambda$-condition in general and hence one cannot define a monoidal product in $\mathbf{QDesign}_{B}$ in general.

\end{document}